\documentclass[12pt]{article}

\usepackage[top=1in,bottom=1in,left=1in,right=1in]{geometry}

\title{\vspace{-1cm}  A unified nonparametric fiducial approach to interval-censored data}

\author{
Yifan Cui\thanks{Department of Statistics and Data Science, National University of Singapore}~
Jan Hannig\thanks{Department of Statistics and Operations Research, UNC-Chapel Hill}~
Michael Kosorok\thanks{Department of Biostatistics \& Department of Statistics and Operations Research, UNC-Chapel Hill}}
\date{}

\usepackage{amsthm,amsmath,amssymb,bbm,bm}
\usepackage{natbib}
\usepackage{multirow}
\usepackage[pdftex]{graphicx}
\usepackage{wrapfig}
\usepackage{tikz}
\usepackage{centernot}
\usepackage{array}
\usepackage{url}
\usepackage{algorithmic}
\usepackage[linesnumbered, ruled, vlined]{algorithm2e}
\usepackage{mathrsfs}
\usepackage{dsfont}
\usepackage{titling}
\usepackage{relsize}
\usepackage{rotating}
\usepackage{enumitem}
\usepackage{titling}
\usepackage{float}
\usepackage{subcaption}
\usepackage{esvect}
\usepackage[font=small,labelfont=bf]{caption}

\newtheorem{assumption}{Assumption}

\newtheorem{theorem}{Theorem}[section]
\newtheorem{lemma}{Lemma}[section] 
\newtheorem{proposition}{Proposition}[section]

\newtheorem{example}{Example}

\newcommand{\yifan}[1]{\textcolor{black}{#1}}

\newcommand{\bu}{\vv{u}}
\newcommand{\bv}{\vv{v}}
\newcommand{\bU}{\vv{U}}
\newcommand{\bl}{\vv{l}}
\newcommand{\br}{\vv{r}}
\newcommand{\bL}{\vv{L}}
\newcommand{\bR}{\vv{R}}
\newcommand{\bt}{\vv{t}}

\newcommand{\tunif}{{\text{Unif}}}

\newcommand{\tbeta}{{\text{Beta}}}
\newcommand{\texp}{{\text{Exp}}}
\newcommand{\tgamma}{{\text{Gamma}}}

\begin{document}
\maketitle
\vspace{-2cm}
\abstract{
 Censored data, where the event time is partially observed, are
challenging for survival probability estimation.
In this paper, we introduce a novel nonparametric fiducial approach to interval-censored data, including right-censored, current status, case II censored, and mixed case censored data.
The proposed approach leveraging a simple Gibbs sampler has a useful property of being ``one size fits all'', i.e., the proposed approach automatically adapts to all types of non-informative censoring mechanisms.
As shown in the extensive simulations, the proposed fiducial confidence intervals significantly outperform existing methods in terms of both coverage and length.
In addition, the proposed fiducial point estimator has much smaller estimation errors than the nonparametric maximum likelihood estimator.
Furthermore, we apply the proposed method to Austrian rubella data and a study of hemophiliacs infected with the human immunodeficiency virus. The strength of the proposed fiducial approach is not only estimation and uncertainty quantification but also its automatic adaptation to a variety of censoring mechanisms.}

\vspace{0.1cm}

\noindent {\bf keywords:}
 Censored data, Current status data, Fiducial inference, Gibbs sampler, Mixed case censoring, Survival analysis

\section{Introduction}\label{sec:intro}
Censored survival data are \yifan{ubiquitous} in biomedical studies when actual clinical outcomes, such as death, disease recurrence, or distant metastasis, may not be directly observable for such reasons as periodic follow-up and early dropout.
\yifan{Interval-censored data arise when a random variable of interest can not be observed,
but can only be determined to lie in an interval obtained from a sequence
of inspection times, i.e., a failure time $T$ is known only to lie within an interval $I=(L,R]$,} which is more challenging than right-censored data because much less information is contained in such intervals.
\yifan{We refer to \cite{huang1997interval,sun2007statistical} for a comprehensive review on interval-censored data and 
\cite{jacobsen1995} who treat interval-censoring as a special case of coarsening at random.}
One extreme case is current status data where the survival status of a subject is inspected at a
single random monitoring time, thus yielding an extreme form of interval-censoring.

\cite{peto1973} proposed a nonparametric maximum likelihood estimation (NPMLE) of the survival function for interval-censored data using the Newton-Rapshon algorithm.
\cite{turnbull1976empirical} showed that the NPMLE of the survival distribution is only unique up to a set of intervals, which may be called the innermost intervals, also known as the Turnbull intervals or the regions of the maximal cliques. \cite{turnbull1976empirical} then suggested a self-consistent expectation maximization to compute the maximum likelihood estimators.
While there is no closed form representation of the NPMLE based on general interval-censored data, an NPMLE
with a convergence guarantee was developed in \cite{groeneboom1992information} for current status and case II censoring data. 
\cite{wellner1995interval} studied the consistency of the NPMLE
where each subject gets exactly $K$ examination times.
The consistency of the NPMLE under the mixed case censoring has been studied by \cite{van2000preservation,anton2000}.

Constructing pointwise confidence intervals for the distribution function of $T$ at a given time $t$ is more challenging in a general interval-censoring setting. It is known that bootstrapping from the NPMLE of the distribution function is inconsistent for both the current status and case II censoring models \yifan{\citep{kosorok2008bootstrapping,sen2010inconsistency,sen2015model}}.
\cite{banerjee2005confidence} proposed likelihood ratio-based confidence intervals for current status data.
Furthermore, \cite{sen2007pseudolikelihood} proposed a pseudo-likelihood approach to mixed case censoring data, which may not be as efficient as the NPMLE and does not achieve nominal coverage.
While the $m$-out-of-$n$ bootstrap \citep{lee2006} and subsampling methods \citep{politis1999subsampling} are consistent, and the corresponding confidence intervals achieve nominal coverage, their resulting confidence intervals are too wide. It is also important to note that, to our knowledge, all previous methods,  including \cite{banerjee2005confidence,lee2006,sen2007pseudolikelihood,politis1999subsampling,sen2015model}, require the choice of tuning parameters such as the block size.

This paper introduces a novel nonparametric fiducial approach to interval-censored data. Fiducial inference can be traced back to a series of articles by
R. A. Fisher \citep{Fisher1930,Fisher1933} who introduced the concept as a potential replacement of
the Bayesian posterior distribution.
\yifan{
 Posterior distribution was at the beginning of 20th century called ``inverse probability'' in contrast to ``direct probability'', better known as likelihood \citep{Fisher1922}.
From a mathematical point of view, the difference between the posterior and fiducial distributions is due to the way the distribution on the parameter space is defined. The former uses a conditional probability which requires selecting a prior probability. The latter transports the probability distribution from a given distribution on an auxiliary space using a measurable function, which we call the data generating equation. 
In both cases, we end up with a probability measure that is not uniquely determined by the likelihood as a change of prior in one case, and the data generating equation in the other can affect the resulting answer.} 
Other related approaches include Dempster-Shafer theory \citep{Dempster:1968vd,shafer1976mathematical}, inferential models \citep{MartinLiu2013a,martin2015inferential}, confidence distributions \citep{SinghXieStrawderman2005, XieSingh2013, HjortSchweder2018}, and objective Bayesian inference \citep{BergerBernardoSun2009, BergerBernardoSun2012}. Many additional references can be found in \cite{XieSingh2013}, \cite{schweder2016confidence},  \cite{hannig2016generalized}, and \cite{cui2021confidence}.
  Since the mid 2000s, there has been renewed interest in modifications of fiducial inference.
  \cite{WangYH2000, TaraldsenLindquist2013} showed how fiducial distributions naturally arise within a decision theoretical framework.
  \cite{hannig2016generalized} formalized the mathematical definition of generalized fiducial distribution. Having a formal definition allowed the application of fiducial inference to other fields, such as psychology \citep{liu2016generalized,liu2017generalized,Liu2019,Neupert2019} and forensic science \citep{HANNIG2019572}.
\yifan{\cite{cuihannig2019} considered a nonparametric fiducial approach to right-censored data which is a special type of interval-censored data.
Their method does not use a Gibbs sampler and applies  only to right-censored data. 
}

The proposed fiducial approach implemented by a simple Gibbs sampler has a useful property of being``one size fits all'', i.e., the proposed approach automatically adapts to all types of non-informative censoring:

(i) Exact data: $L=T, R=T$;

(ii) \yifan{Right-censored data}: $R=\infty$ \yifan{for right-censored observations;}

(iii) \yifan{Left-censored data}: $L=0$ \yifan{for left-censored observations;}

(iv) Case I censoring (current status data): \yifan{only one inspection time is available}, i.e., either $L=0$ or $R=\infty$;


(v) Case $K$ censoring: $K$ observation/inspection times \yifan{$C_1,\ldots,C_K$ with observations $L,R\in\{0,C_1,\ldots,C_K,\infty\}$; $K$ might tend to
infinity as the sample size tends to infinity \citep{lawless2006models}}.

(vi) Mixed case censoring: an arbitrary number of observation/inspection times, i.e., a mixture of the above censoring mechanisms.

We use the fiducial distribution to construct a point estimator and pointwise confidence \mbox{intervals} for the distribution function.
\yifan{In this paper, we perform extensive simulations following the configurations considered in the previous literature \citep{banerjee2005confidence,sen2007pseudolikelihood}. In these simulations, the proposed confidence interval maintains coverage in situations where most existing methods have coverage problems, and meanwhile, it has the shortest length
among all the confidence intervals. In addition, the proposed fiducial point estimator has the smallest
mean squared error compared to various NPMLE estimators.}
Furthermore, we apply the proposed approach to Austrian rubella data and a study of hemophiliacs infected with the human immunodeficiency virus.

The advantages of the proposed fiducial approach are two-fold: 1) The proposed fiducial distribution and corresponding algorithm adapt to a variety of censoring mechanisms automatically, which is a substantial advantage when information about the inspection times is not available. This is somewhat of an art, and our contribution appears valuable for such scientific applications;
 2) In our simulations, the proposed fiducial approach significantly outperforms existing methods in terms of both point estimators and confidence intervals.

\section{Methodology}\label{sec:method}

\subsection{Setup and notation}
Before describing the proposed fiducial approach, we first introduce some general notation. 
Suppose the observed data are $\{I_i=(l_i,r_i], i=1,\ldots, n\}$.
For the censoring mechanism, we consider non-informative censoring \citep{oller2004interval} under which intervals do not provide any further information than the fact that the event time lies in the interval:
\begin{equation}\label{eqr:non_dependence}
\Pr(T\leq t| L=l,R=r, L<T\leq R) = \Pr(T\leq t|l<T\leq r).
\end{equation}
\yifan{We refer to \cite{sun2007statistical,kalbfleisch2011statistical} for the possibility of including  covariates. Suppose} we are interested in the unknown distribution function \yifan{$F(t)$ of the survival time $T$} at time $t$.
\cite{law1992effects} showed through simulations that treating observations as right-censored data after a midpoint imputation does not preserve type I error. There is clearly a need for methods specifically designed for interval-censored data.

\subsection{A data generating equation perspective}\label{sec:2.2}

In this section, we first explain the definition of a fiducial distribution and then demonstrate how to apply it to interval-censored data. \yifan{This derivation will be conditional on the observed $l_i,r_i,\ i=1,\ldots, n$.  The common assumption~\eqref{eqr:non_dependence} allows us to ignore the potential dependence between $T$ and $L,R$, and treat the observed $l_i,r_i$ as fixed. We provide an alternative derivation of the same generalized fiducial distribution in Appendix~\ref{app:censoring}, where we explicitly model the relationship between failure and censoring times treating $L$ and $R$ as random.
}

We start by expressing the event times $T_i$ using
\begin{equation}\label{eq:DGE}
T_i=F^{-1}(U_i),\quad i=1,\ldots n,
\end{equation}
where $U_i$ are independent $\tunif(0,1)$ and $F^{-1}(u) = \inf\{t : F(t)\geq u \}$. 

Recall that we do not observe the exact values of $T_i$ but instead observe their lower and upper bounds $(l_i,r_i]$.
By a simple calculation,
\begin{align*}
F^{-1}(u_i) > l_i \text{~if and only if~} F(l_i) < u_i,\\
F^{-1}(u_i) \leq r_i \text{~if and only if~} F(r_i) \geq u_i.
\end{align*}
Consequently, the inverse of the data generating equation \eqref{eq:DGE} expressed by the observed data is
\begin{equation}
Q_{\bl,\br}(\bu) =
\{F\, :\ \, F(l_i)< u_i\leq F(r_i), i=1,\ldots,n\}.  \label{eq:constraint}
\end{equation}
Note that $Q_{\bl,\br}(\bu)$ is a set of cumulative distribution functions. By Lemma~\ref{lemma:equiv2} provided in Appendix~\ref{sec:A},
$Q_{\bl,\br}(\bu)\neq\emptyset$ if and only if $\bu$ satisfy:
\begin{equation}\label{eq:constraint2}
\mbox{whenever $r_i \leq l_j$ then $u_i < u_j$.}
\end{equation}

A fiducial distribution is obtained by inverting the data generating equation, 
i.e.,
the distribution of
 $Q_{\bl,\br}(\bU^\star)$,
where $\bU^\star$ is the uniform distribution on the set
$\{\bu^\star : Q_{\bl,\br}(\bu^\star)\neq\emptyset\}$.
\yifan{The random functions defined for each $t$ and $\bU^\star$
\[
F^U(t) \equiv \min\{U^\star_i, \mbox{ for $i$ such that } t < L_i\},
\]
and
\[
F^L(t) \equiv \max\{U^\star_i, \mbox{ for $i$ such that } t\geq R_i\},
\]
where $\min\emptyset =1$ and $\max\emptyset =0$, are non-decreasing and right continuous.
Note that any distribution function lying between $F^U$ and $F^L$ is an element of the closure, in weak topology, of $ Q_{\bl,\br}(\bU^\star)$.
Thus, the functions $F^U$ and $F^L$ will be called the upper and lower fiducial bounds throughout.}

\subsection{A simple Gibbs sampler}\label{sec:Gibbs}
\yifan{In this section, we propose a novel Gibbs sampler to efficiently sample $\bU^\star$.
A sample from the fiducial distribution obtained from the Gibbs sampler can then be used to form a point estimator and confidence intervals for the unknown distribution function $F(t)$ in the same way that posterior samples are used in the Bayesian context.}

We need to generate $\bU^\star$ from the standard uniform distribution on a set described by Equation~\eqref{eq:constraint2}.
We achieve this by a simple Gibbs sampler.
For each fixed $i$, we denote the random vector $\bU^\star$ with the $i$-th observation removed by $\bU^\star_{[-i]}$. If  $\bU^\star$ satisfies the constraint \eqref{eq:constraint2}, so does $\bU^\star_{[-i]}$. The proposed Gibbs sampler is based on the conditional distribution of $ U^\star_i\mid \bU^\star_{[-i]}$,
which is a uniform distribution on a segment determined by \eqref{eq:constraint2}. The details are described in Algorithm~\ref{alg:gibbs}. The proposed Gibbs sampler requires starting points.
We randomly sample $\bu_0$ from independent $\tunif$(0,1) and sort $\bu_0$ according to the order of $(\bl+\br)/2$ as initial points.

\vspace{0.5cm}
\begin{algorithm}[H] 
\SetAlgoLined
\caption{A simple pseudo-algorithm for the fiducial Gibbs sampler} \label{alg:gibbs}
\KwIn{Dataset  $(\bl,\br)$,
$n_{\text{mcmc}}$, $n_{\text{burn}}$, and vector $\bt_{\text{grid}}$ of length $m$.}
\ShowLn Sample $\bu_0$ from independent $\tunif$(0,1)\;
\ShowLn Sort $\bu_0$ according to the order of $(\bl+\br)/2$\;
\ShowLn Run Gibbs sampler using $\bu_0$ as initial values\;
\For{$j = 1$ \textbf{to} $n_{\text{burn}}+n_{\text{mcmc}}$}
{
\For{$i = 1$ \textbf{to} $n$}
{
\ShowLn $\bu_j^\ast=\bu_{j-1}[-i],\bl_j^\ast=\bl[-i], \br_j^\ast=\br[-i]$\;
\ShowLn $v$=which($\bl[i]\geq \br_j^\ast$),~ $a = \max(\bu_j^\ast[v],0)$\;
\ShowLn $w$=which($\br[i]\leq \bl_j^\ast$),~ $b = \min(\bu_j^\ast[w],1)$\;
\ShowLn $\bu_j[i]=\tunif(a,b)$\;
}
Regrid $\bu_j$ according to $\bt_{\text{grid}}$, and denote the lower and upper bounds by $\bu_j^L$ and $\bu_j^U$, respectively\;
}
\ShowLn Interpolate through $\bu^L$ and $\bu^U$ for each Markov chain Monte Carlo sample after burn-in\;
\For{$k = 1$ \textbf{to} $n_{\text{mcmc}}$}{
\vspace{-0.1cm}
\begin{equation*}
\begin{aligned}
& \underset{\bu_k^I}{\text{minimize}}
& &  \sum_{i=1}^{m+1} (u_i-u_{i-1})^2\\ 
& \text{subject to}
& & \bu_k^I \leq \bu_k^U,~~\bu_k^I \geq \bu_k^L
\end{aligned}
\end{equation*}
 where $\bu_k^I= (u_1,\ldots,u_m)^T$, $u_0$, and $u_{m+1}$ are sampled from the independent $\tbeta(1/2,1/2)$ transformed to $(0,\bu^U[1])$ and $(\bu^L[m],1)$, respectively.
}
\Return  the fiducial samples $\bu_k^{U}$, $\bu_k^{L}$, $\bu_k^I$ evaluated on $\bt_{\text{grid}}$, $k=1,\ldots,n_{\text{mcmc}}$.
\label{algorithm}
\end{algorithm}

\yifan{Using this algorithm we generate a fiducial sample $\bU^\star_k,\ k=1,\ldots,n_{\text{mcmc}}$. Based on the fiducial sample,} 
we construct two types of pointwise confidence intervals  \yifan{by finding intervals of a given fiducial probability. Similar to \cite{cuihannig2019},  we define} conservative and linear interpolation intervals, using appropriate quantiles of fiducial samples.
In particular, a $95\%$ conservative confidence interval  is formed by taking the empirical 0.025 quantile of $F_k^{L}(t)$ as a lower limit and the empirical 0.975 quantile of $F_k^{U}(t)$ as an upper limit \citep{Dempster2008, shafer1976mathematical}.
An alternative pointwise interval is based on selecting a suitable representative of \yifan{each $ Q_{\bl,\br}(\bU^\star_k)$}.
We propose to fit a continuous distribution function \yifan{$F^I_k(t)$} by using linear
interpolation via a quadratic programming.
\yifan{The details of the algorithm are provided in Algorithm~1.}
Thus, a 95\% linear interpolation confidence interval for $F(t)$ is formed by using the empirical 0.025 and 0.975 quantiles of $F_k^{I}(t)$.
Finally, we propose to use the pointwise median of the linear interpolation fiducial samples $\{F_k^I(t),k=1\ldots,n_{\text{mcmc}}\}$  as a point estimator for the distribution function. 
Hereinafter, we refer to fiducial confidence intervals as linear interpolation fiducial confidence intervals as we recommend the linear interpolation fiducial samples for practice.

\subsection{Further illustration with two simulated examples}\label{sec:2.3}
To demonstrate the proposed fiducial approach, we present two toy examples in this section. In the first setting, \yifan{current status data,} suppose that the event time \yifan{$T$} and observation time \yifan{$C_1$} both follow the exponential distribution $ \texp (1)$ \citep{banerjee2005confidence,sen2007pseudolikelihood}.
In the second setting \yifan{with case II censoring}, the event time follows a $\tgamma(2,1)$ distribution, and observation times \yifan{$C_1, C_2$} are taken to be $\tunif(0,2)$ and $C_1 + 0.5 + \widetilde C_1$, respectively, with $\widetilde C_1$ independent of $C_1$ and also following $\tunif(0,2)$ \citep{sen2007pseudolikelihood}.

\yifan{In Figure~\ref{newtoy} we show two fiducial samples from the fiducial distribution
$Q_{\bl,\br}(\bU^\star)$ for a small dataset ($n$ = 20) for the first setting, the current status data.
The dashed curves are the lower and upper fiducial bounds, and the solid curve is the corresponding linear approximation.
The crosses correspond to the observations of the type $(0,r_i]$, where on the horizontal axis we plot the $r_i$ and on the vertical axis we show the corresponding $U_i^\star$.
The circles are the observations of the type $(l_j,\infty]$, where on the horizontal axis we plot the $l_j$ and on the vertical axis we again show the $U_j^\star$. Note that the upper fiducial bound has jumps only at values corresponding to some of the circles, with the rest of the circles being above the upper fiducial bound. Similarly, the lower fiducial bound jumps only at locations corresponding to some of the crosses with the rest of the crosses being below it.}

\begin{figure}[h]
\begin{center}
\includegraphics[width=5.5cm]{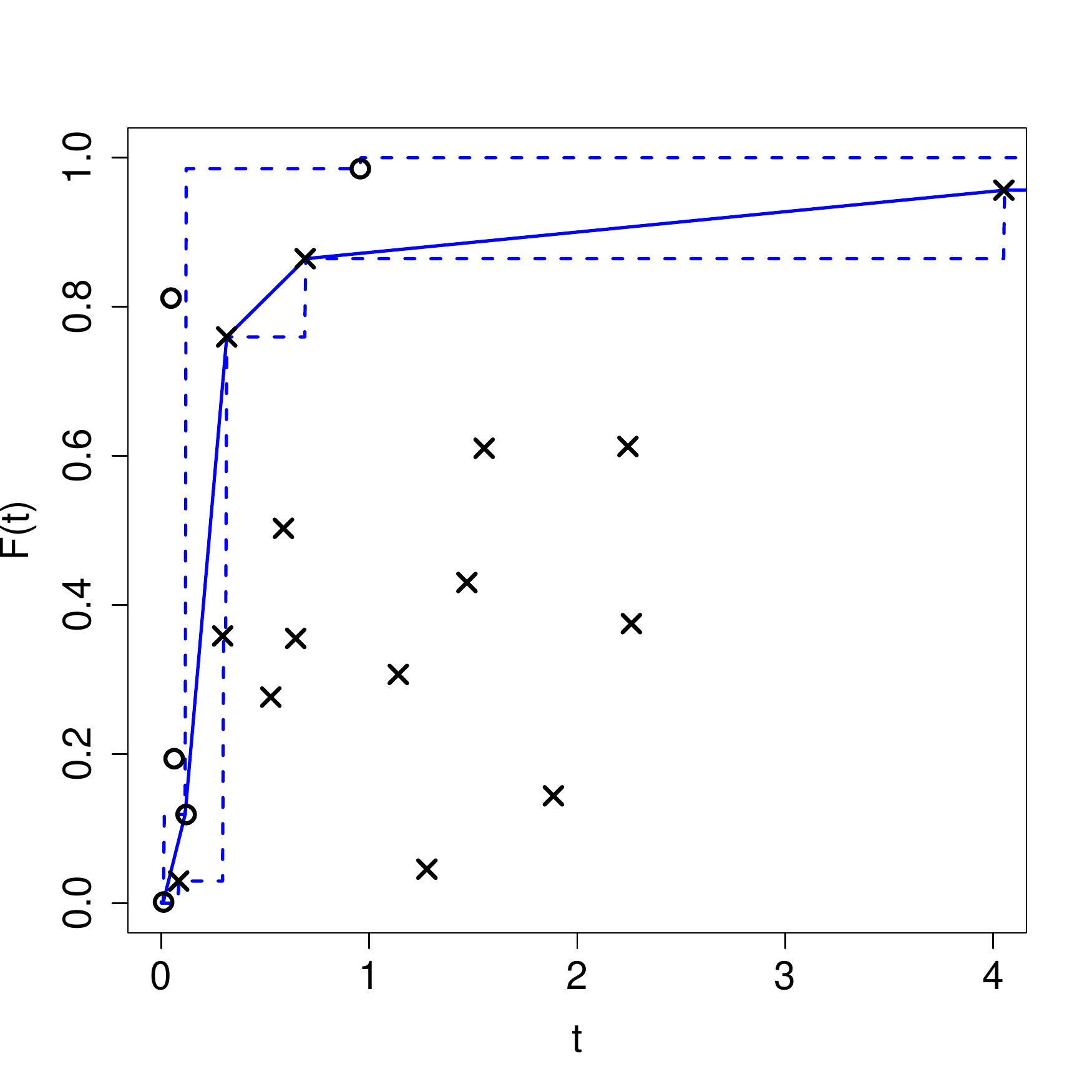}
\includegraphics[width=5.5cm]{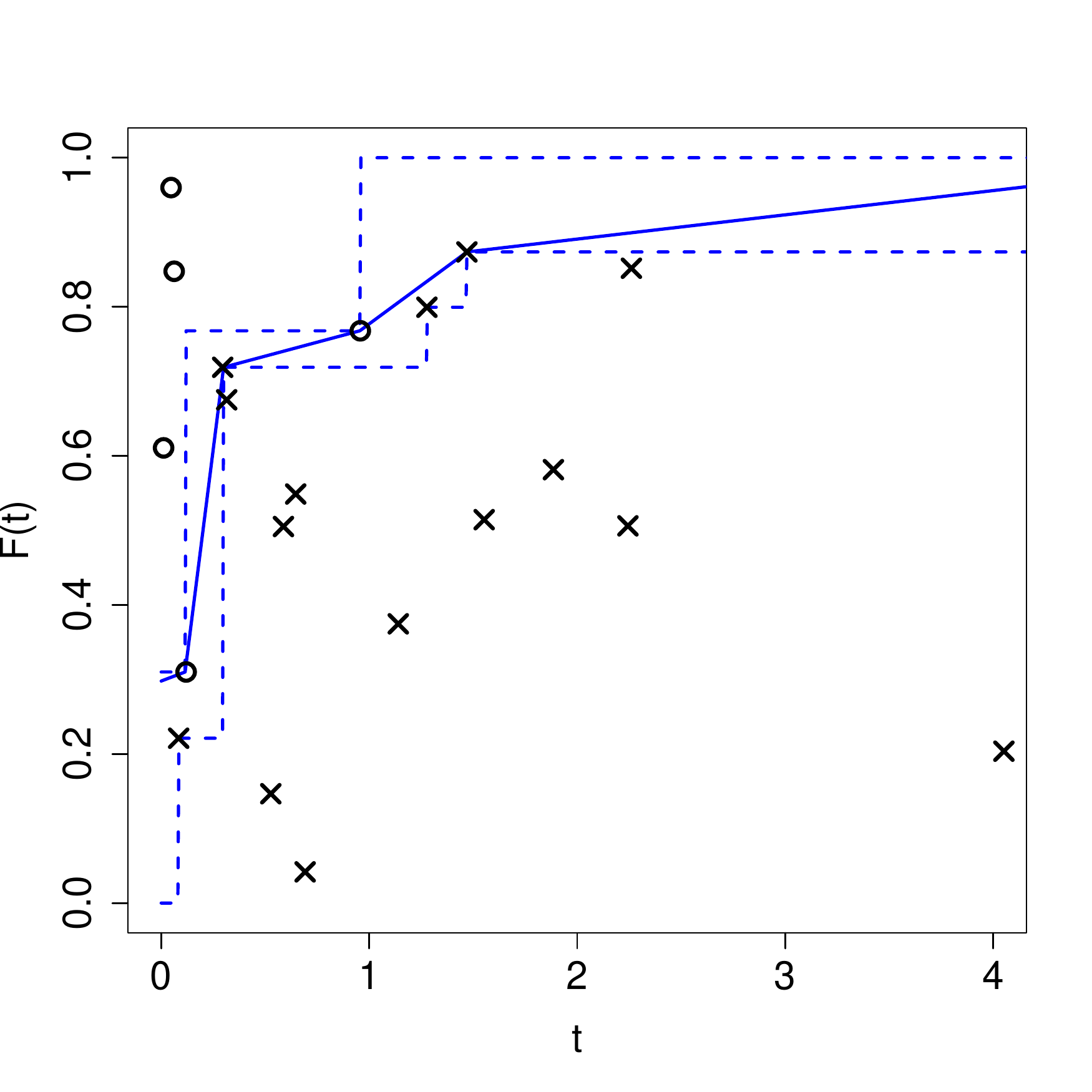}
\end{center}
\caption{\yifan{Two fiducial samples from the fiducial distribution for a small dataset ($n$ = 20) for the first setting, the current status data; the event time and observation time both follow $ \texp (1)$. The crosses correspond to the observations of the type $(0,r_i]$, where on the horizontal axis we plot the $r_i$ and on the vertical axis we show the corresponding $U_i^\star$.
The circles are the observations of the type $(l_j,\infty]$, where on the horizontal axis we plot the $l_j$ and on the vertical axis we again show the $U_j^\star$.}}
\label{newtoy}
\end{figure}

Next, we consider the sample size of the simulated data $n=200$. The fiducial estimates were based on 1000 iterations after 100 burn-in times.
The left panels of Figures~\ref{fig:1} and \ref{fig:2} present the last Markov chain Monte Carlo sample of the lower and upper fiducial bounds as well as the linear interpolation fiducial sample.
As the fiducial distribution reflects the uncertainty, we do not expect every single fiducial curve to be close to the true cumulative distribution function.
Furthermore, the right panels of Figures~\ref{fig:1} and \ref{fig:2} present 95\% linear interpolation confidence intervals and corresponding point estimators, respectively.

\begin{figure}[H]
\begin{center}
\includegraphics[width=5.5cm]{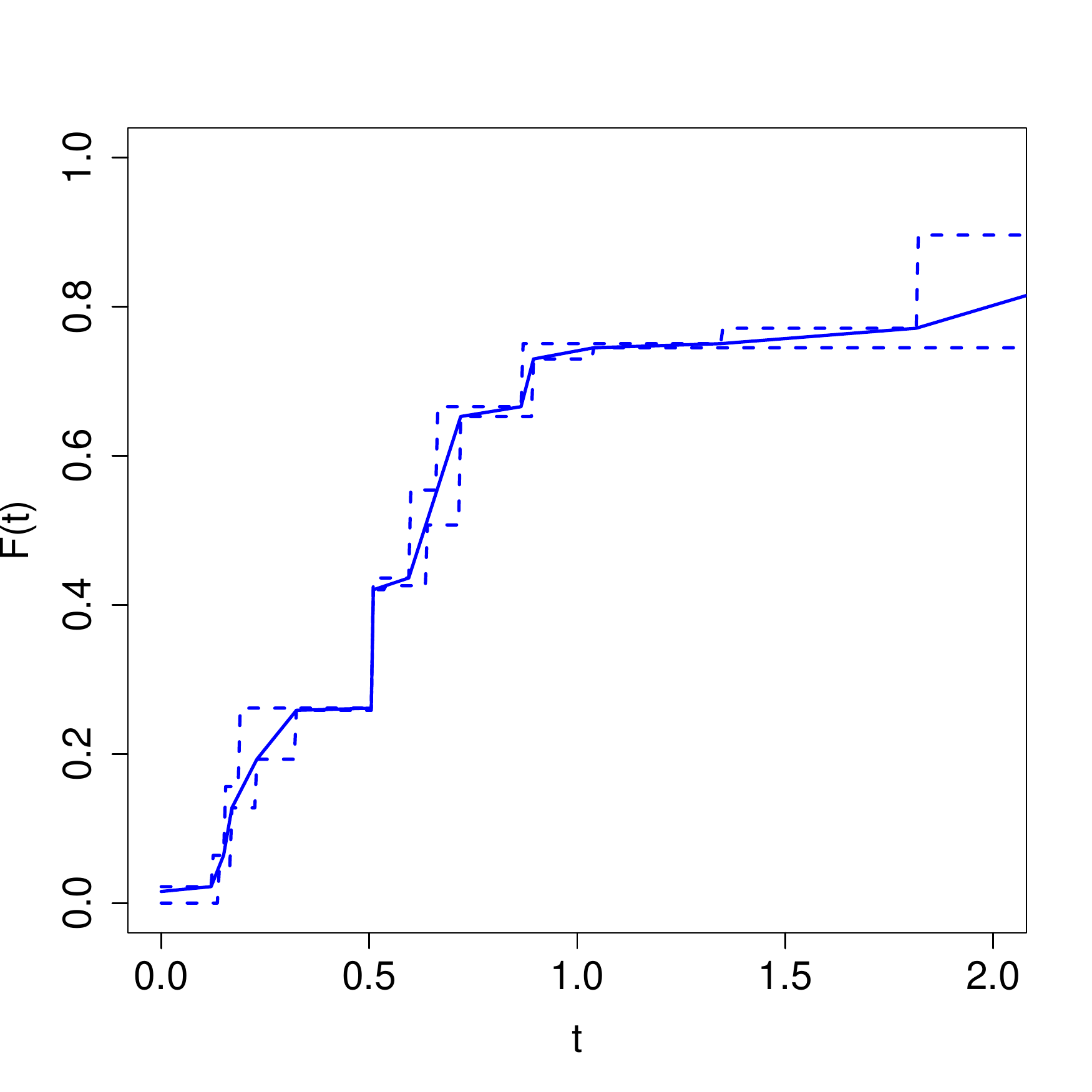}
\includegraphics[width=5.5cm]{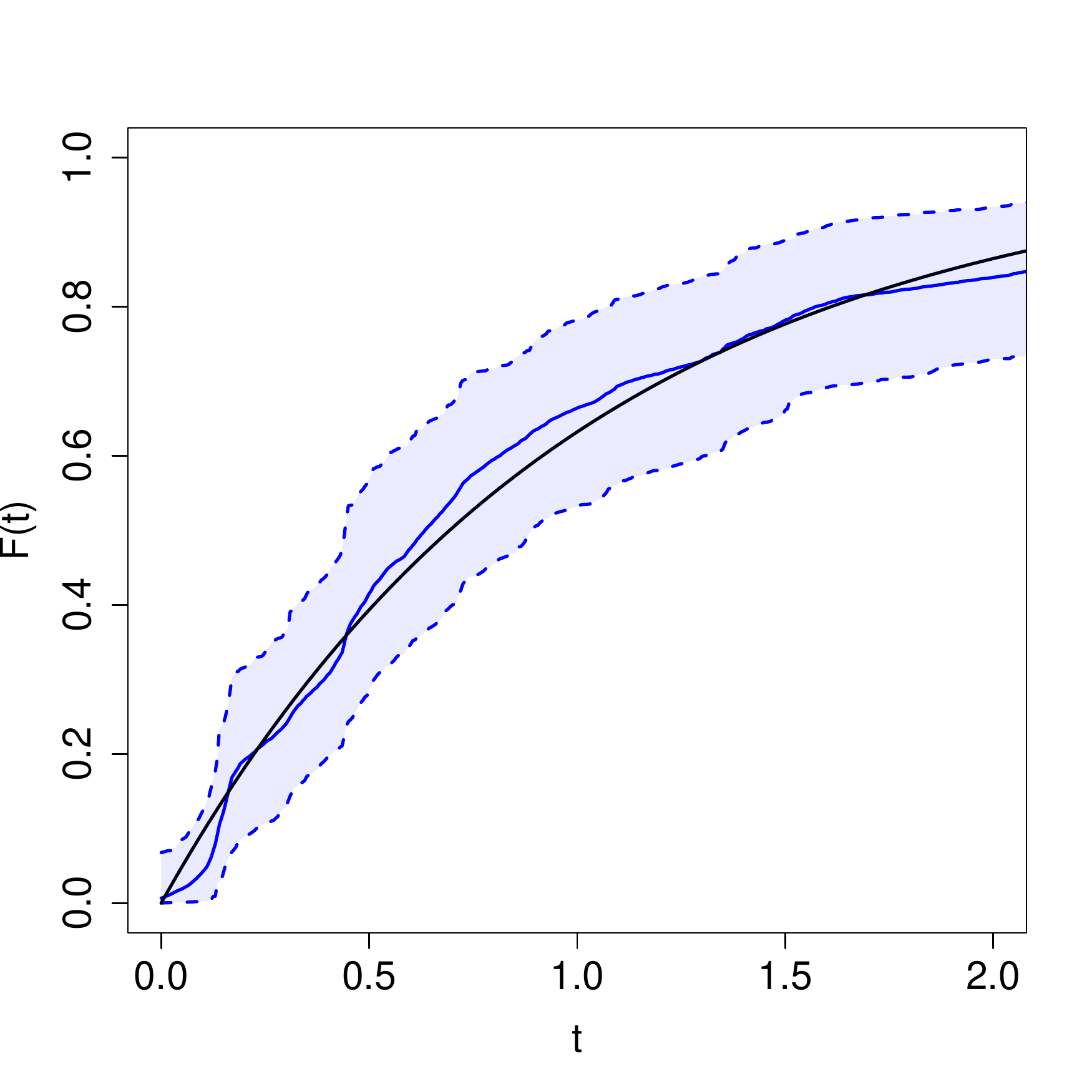}
\end{center}
\caption{Setting 1 (current status data): the event time and observation time both follow $ \texp (1)$. Left panel: The last Markov chain Monte Carlo sample from fiducial distribution. The dashed curves are the realizations of the lower fiducial sample $F^{L}(t)$ and upper fiducial sample $F^{U}(t)$, respectively. The solid line is the linear interpolation $F^{I}(t)$. Right panel: true cumulative distribution function (black line), 95\% confidence interval (dashed blue line) and corresponding point estimator (solid blue line).}
\label{fig:1}
\end{figure}

\begin{figure}[H]
\begin{center}
\includegraphics[width=5.5cm]{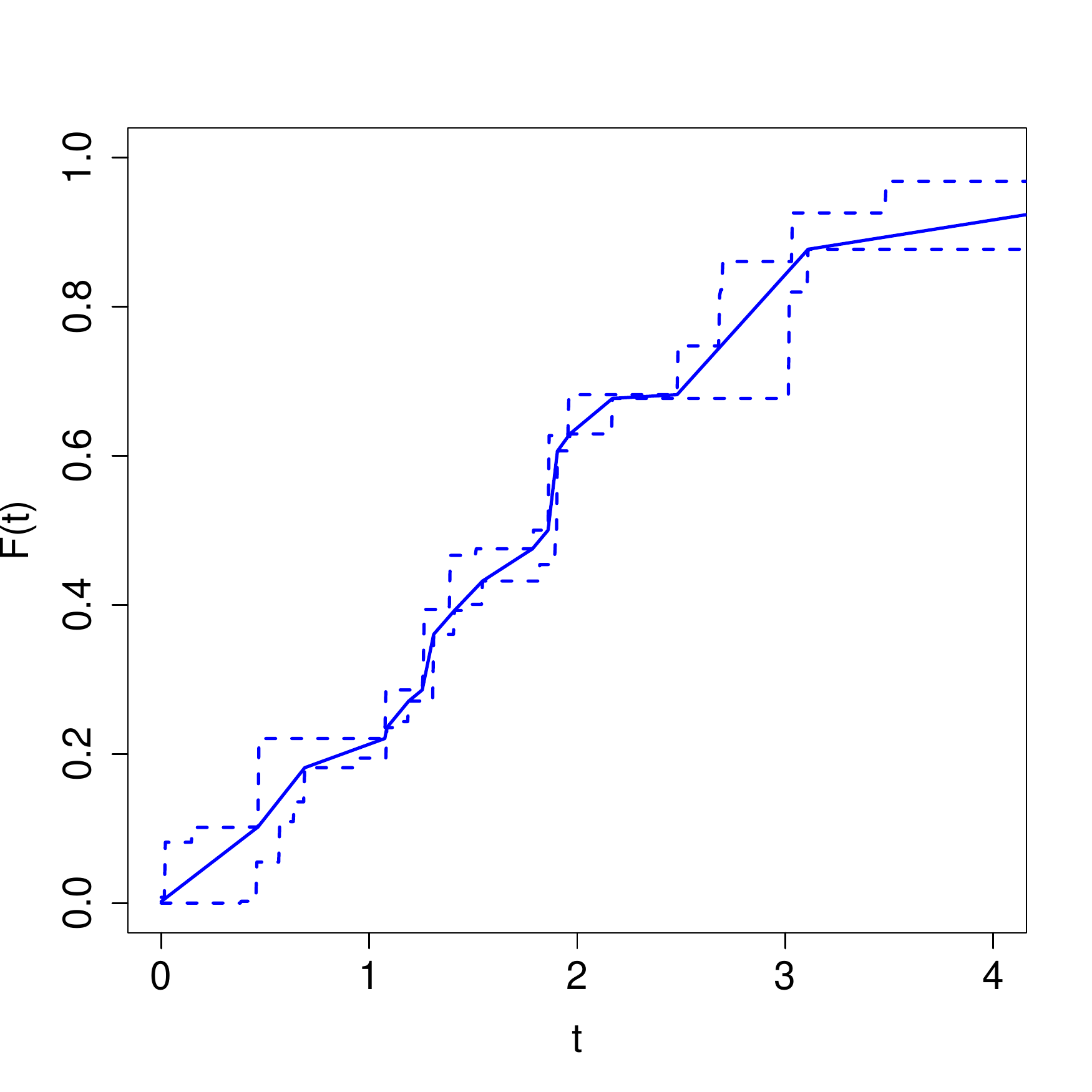}
\includegraphics[width=5.5cm]{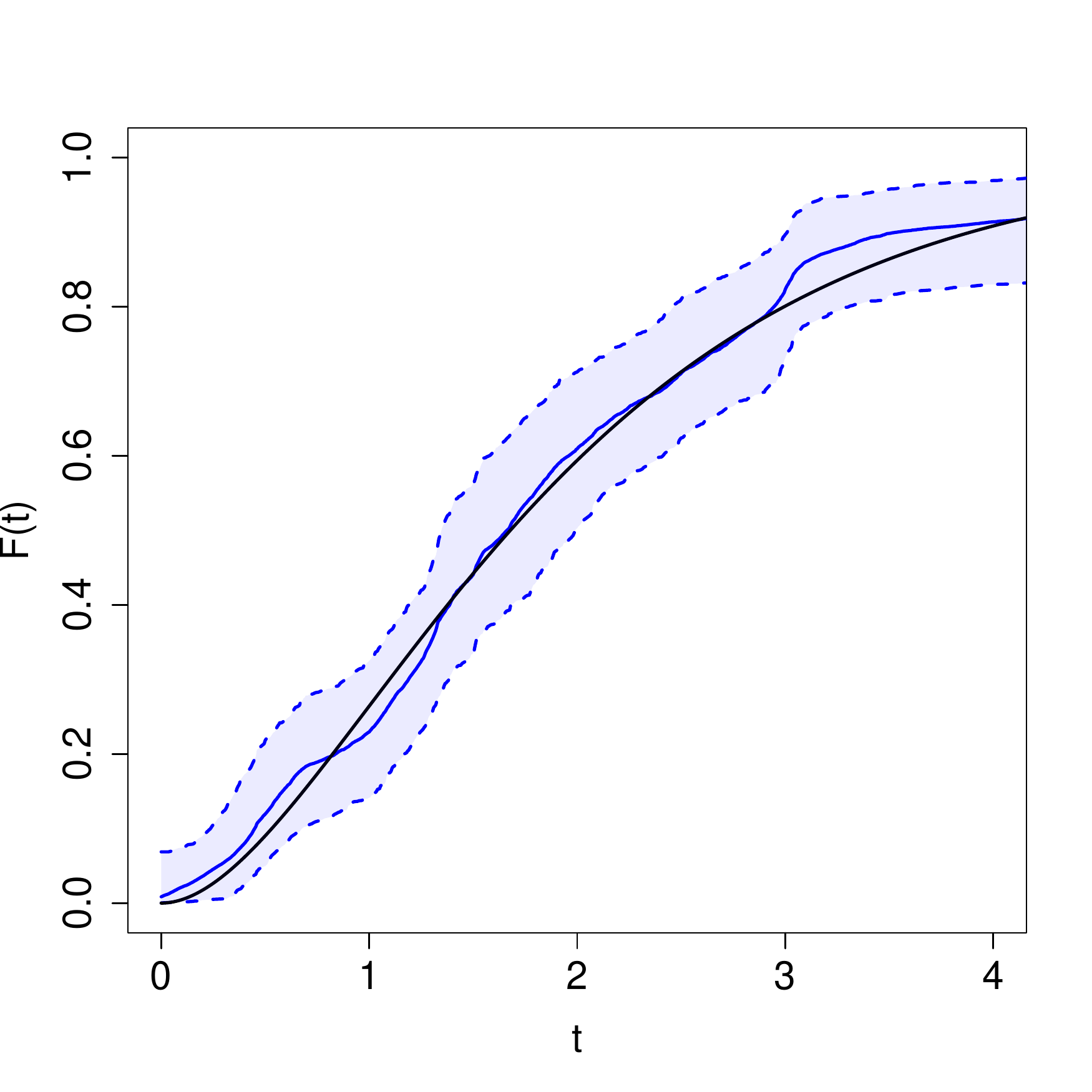}
\end{center}
\caption{Setting 2 (case II censoring): the event time follows $\tgamma(2,1)$, and observation times \yifan{$C_1, C_2$} are $\tunif(0,2)$ and $C_1 + 0.5 + \widetilde C_1$, respectively, with $\widetilde C_1$ independent of $C_1$ and also following $\tunif(0,2)$. Left panel: The last Markov chain Monte Carlo sample from fiducial distribution. The dashed curves are the realizations of the lower fiducial sample $F^{L}(t)$ and upper fiducial sample $F^{U}(t)$, respectively. The solid line is the linear interpolation $F^{I}(t)$. Right panel: true cumulative distribution function (black line), 95\% confidence interval (dashed blue line) and corresponding point estimator (solid blue line).}
\label{fig:2}
\end{figure}

\section{Theoretical results}\label{sec:theory}

\subsection{Connection to the nonparametric maximum likelihood estimator}

\yifan{
 In Section~\ref{sec:theory}, we present the theoretical results in two directions. First, we show that the mode of the fiducial distribution is the NPMLE. We found this result surprising as fiducial distribution is not the same as normalized likelihood. 
 The other direction is asymptotic analysis of the case $K$ censoring where the number of inspection times $K$ goes to infinity, which will be studied in the next subsection.
 It is well known that counting process techniques that have been successfully used in asymptotic analysis of right-censored data cannot be used for fixed $K$ interval-censored data where theoretical results appear to be much harder and will not be studied here.
 }

\begin{proposition}\label{theorem:consistency}
\yifan{Assuming~\eqref{eqr:non_dependence}, for a given dataset $(\bl,\br)$ any $F$ maximizing fiducial probability $\text{Pr}^\star(F\in Q_{\bl,\br}(\bU^\star))$ is an NPMLE}.
\end{proposition}

Proposition~\ref{theorem:consistency} provides some justification for the fiducial approach. \yifan{Note that the $\text{Pr}^\star(F\in Q_{\bl,\br}(\bU^\star))$ is called {\em plausibility} in the Dempster-Shafer theory \citep{shafer1976mathematical}, so the result could be interpreted as maximum plausibility and maximum likelihood agree in this model.
This result suggests a possible way to create a simultaneous fiducial confidence interval as a ball of $1-\alpha$ fiducial probability with its center being the most plausible distribution function, the NPMLE. In practice, this would be constructed by selecting a ball that contains $(1-\alpha)$100\% fiducial samples from the Gibbs sampler of Section~\ref{sec:Gibbs}.}

 \subsection{\yifan{ Bernstein-von Mises theorem under Assumption~\ref{as:2}}}
 \yifan{
 Recall that the fiducial distribution is a data-dependent distribution 
which is defined for every fixed
dataset $(\bl,\br)$. It can be made into a random measure 
in the same way as one defines the
usual conditional distribution, i.e., by plugging random variables $(\bL,\bR)$ into the observed data. 
In this section, we study the asymptotic behavior of this random measure under the following condition:}
\begin{assumption}
$    n^2 \max_{i=1,\ldots,n} |R_i-L_i| \to 0$ in probability.
\label{as:2}
\end{assumption}

\yifan{
Assumption~\ref{as:2} provides a sufficient condition for $\sqrt n$-convergence and the Bernstein-von Mises theorem~\ref{main}, which needs the length of each interval to be short.
This assumption can also be viewed as a case $K$ censoring where the number of inspection times $K$ goes to infinity at a certain rate.
A similar but different asymptotic assumption for interval-censoring is Assumption (A1) in \cite{huang1999} which basically requires enough exact observations. 
Both assumptions essentially impose the restriction that the censored data are in some sense close to the uncensored data.
}

Next, we prove a central limit theorem  for $F^L(t)$. The same result holds for $F^U(t)$.
\begin{theorem}\label{main}
Suppose the true cumulative distribution function is absolutely continuous.
If Assumption~\ref{as:2} holds,
\begin{align}
n^{1/2}\{ F^L(\cdot)-\widehat F(\cdot)  \} \rightarrow B_F(\cdot),
\label{eq:gaussian}
\end{align}
in distribution on Skorokhod space $D[0,\tau]$ in probability, where $\widehat F(\cdot)$ is the empirical cumulative distribution function constructed based on \yifan{the unobserved failure times $T_i$, $\tau$ is the end of the follow-up time with} $F(\tau)<1$, and $B_F(\cdot)$ is a Gaussian process with  mean zero  and $\operatorname{cov}(B_F(s),B_F(t))=F(t\wedge s)-F(t)F(s)$.
\end{theorem}

The above theorem establishes a Bernstein-von Mises theorem for the fiducial distribution.
To understand this mode of convergence used here, note
that there are two sources of randomness present. One is from the fiducial distribution derived from each fixed dataset. The other is the usual randomness of the data. The
mode of convergence here is in distribution in probability, i.e., the centered and scaled fiducial
distribution viewed as a random probability measure on $D[0, \tau]$ converges in probability to the
distribution of the Gaussian process described in the right-hand side of Equation~\eqref{eq:gaussian}. \yifan{Mathematically speaking,
for all $\epsilon>0$, 
\[
 \text{Pr}(\rho[n^{1/2}\{ F^L(\cdot)-\widehat F(\cdot)\}, B_F(\cdot)]>\epsilon)\to 0,
\]
where $\rho$ is a metric on
the space of probability measures on $D[0, \tau]$ that metrizes weak topology, e.g.,  Dudley's metric \citep{Shorack2017}, and the probability refers to the randomness of the data $L_i,R_i,T_i$.} 

\yifan{
Assumption~~\ref{as:2} may be limited in certain applications.  That being said, we believe that Theorem~\ref{theorem:consistency} might hold more generally. In particular, we conjecture that under this interval-censoring setting whenever there is a $\sqrt n$-convergence of the NPMLE, there is a Bernstein-von Mises theorem for the fiducial distribution.
It would also be interesting to investigate the convergence rate and distributional result of the proposed fiducial distribution for the fixed case $K$ censoring. In general, we do not expect a $\sqrt n$-convergence rate because \cite{groeneboom1992information,groeneboom2008current} proved a cube rate convergence for the NPMLE for current status data and case II censoring.}

\section{Simulation experiments}\label{sec:simulation}

\subsection{Current status data}
\yifan{We examined the coverage and average length of 95\% fiducial confidence intervals for $F(t_0)$, where, following \cite{banerjee2005confidence,sen2007pseudolikelihood}, we select $t_0$ as the median of the failure distribution.}
We considered the following two scenarios from \cite{banerjee2005confidence}, where the first scenario was also considered in the unpublished longer version of \cite{sen2007pseudolikelihood}:

\yifan{Scenario 1: Let the event time $F$ follow $\texp (1)$ and the observation time follow $\texp (1)$. 
}

\yifan{Scenario 2: Let the event time $F$ follow $\tgamma (3,1)$ and the observation time follow $\tunif (0,5)$.}

We chose sample sizes $n= 50, 75, 100, 200, 500, 800, 1000$ following \cite{banerjee2005confidence}.
\yifan{Each scenario was simulated 1000 times.
The fiducial estimates were based on 1000 iterations after 100 burn-in times.}
 For both scenarios, \yifan{the interval [0,5]} was equally divided into 100 intervals as a fiducial grid, \yifan{where fiducial grid refers to the vector $\bt_{\text{grid}}$ defined in Algorithm~1}. The simulation results are listed in Table~1 for each scenario. 
The results of competing 95\% confidence intervals, such as the likelihood ratio-based method, maximum likelihood based method with nonparametric estimation, subsampling-based method, and parametric (Weibull-based) estimation, can be found in \cite{banerjee2005confidence}.

In the tables, \yifan{LR} denotes the error rate that the true parameter is less than the lower confidence limit; \yifan{UR} denotes the error rate that the true parameter is greater than the upper confidence limit.
The two-sided error rate is obtained by adding the values in columns \yifan{LR} and \yifan{UR}. Values less than 2.5\% in individual columns, 5\% in aggregate, indicate good performance. \yifan{WD} is the average width of the confidence interval.
As can be seen from these tables, the proposed fiducial confidence intervals maintain the aggregate coverage and are much shorter than those considered in \cite{banerjee2005confidence}.
Recall that Tables 1-2 in \cite{banerjee2005confidence} show all considered methods have either substantial or minor coverage problems in these settings.

\begin{table}[H]
\begin{center}
{\textbf{Table~1}. Error rates in percent and average width of $95\%$ confidence intervals for $F(t_0)$.}\\
\label{table1}
\begin{tabular}{cccccccccc}
         & \multicolumn{3}{c}{Scenario 1} & \multicolumn{3}{c}{Scenario 2}\\
         & \yifan{LR}      & \yifan{UR}     & \yifan{WD}      & \yifan{LR}      & \yifan{UR}     &   \yifan{WD}      \\
$n$=50  & 1.1    & 2.4   & 0.414   & 1.5    & 3.6   & 0.429   \\
$n$=75   & 1.0    & 1.6   & 0.364    & 1.1    & 2.3   & 0.382    \\
$n$=100   & 1.0    & 1.4   & 0.332   & 0.6    & 3.1   & 0.351    \\
$n$=200   & 1.3    & 1.8   & 0.262   & 0.6    & 1.4   & 0.280   \\
$n$=500   & 0.7     &  1.2   &  0.188  &  1.1  & 0.7  & 0.205   \\
$n$=800   &  0.9  & 1.7  &  0.159 &  0.7   &  1.1  &   0.174 \\
$n$=1000  &  1.6  & 0.9  &   0.146 &  0.7  & 1.1  &  0.159 \\
\end{tabular}\\
\small{LR denotes the error rate that the true parameter is less than the lower confidence limit; UR denotes the error rate that the true parameter is greater than the upper confidence limit; WD is the average width of the confidence interval. The results of prior methods can be found in Tables 1-2 of \cite{banerjee2005confidence}.}
\end{center}
\end{table}

\subsection{Case II and mixed case censoring}

We considered the following two scenarios from \cite{sen2007pseudolikelihood} and their unpublished longer version:

\yifan{Scenario 3 (case II censoring): Let $F$ follow a $\tgamma(2,1)$ distribution, and the first observation time \yifan{$C_1$} is taken to be $\tunif(0,2)$ and the second observation time $C_2$ is taken as $C_1 + 0.5 + \widetilde C_1$, with $\widetilde C_1$ independent of $C_1$ and also following $\tunif(0,2)$. Recall that we take $t_0$ to be the median of the failure time.
}

\yifan{Scenario 4 (mixed case censoring): The event time distribution $F$ is taken to follow $\texp(1)$. The random number of observation times for an individual $K$ is generated from the discrete uniform distribution on the integers $\{1, 2, 3, 4\}$, and, given $K = k$, the observation times $\{C_{i}\}_{i=1}^k$ are chosen as $k$ order statistics from a $\tunif(0,3)$ distribution.} 

Again, we chose sample sizes $n= 50, 75, 100, 200, 500, 800, 1000$.
\yifan{Each scenario was simulated 1000 times.
The fiducial estimates were based on 1000 iterations after 100 burn-in times.}
 For Scenario 3, \yifan{the interval} [0,5] was equally divided into 100 intervals as a fiducial grid, and \yifan{the interval} [0,3] was equally divided into 100 intervals as a fiducial grid for Scenario 4.
 Again, we examined the coverage and average length of the 95\% fiducial confidence intervals for $F(t_0)$. The results of the 95\% confidence intervals for the 
 pseudo-likelihood ratio method, maximum pseudo-likelihood method,
 kernel-based method, and subsampling-based method were reported in \cite{sen2007pseudolikelihood} and their unpublished longer version.

The simulation results are shown in Table~2 for each scenario. Again, we see that the proposed fiducial confidence intervals maintain the aggregate coverage and are much shorter than those considered in \cite{sen2007pseudolikelihood}. Recall that Tables 2-3 in \cite{sen2007pseudolikelihood} show all considered methods have coverage problems in these settings except for the subsampling-based method.

\begin{table}[H]
\begin{center}
{\textbf{Table~2}. Error rates in percent and average width of $95\%$ confidence intervals for $F(t_0)$.}\\
\begin{tabular}{cccccccccc}
         & \multicolumn{3}{c}{Scenario 3} & \multicolumn{3}{c}{Scenario 4} \\
         & \yifan{LR}      & \yifan{UR}     & \yifan{WD}      & \yifan{LR}      & \yifan{UR}     & \yifan{WD}      \\
$n$=50   & 2.1    & 2.9   & 0.324   & 1.0    & 4.0   & 0.373    \\
$n$=75   & 1.4    & 1.3   & 0.280   & 1.0    & 3.7   & 0.323   \\
$n$=100   & 1.6    & 1.3   & 0.252   & 1.4    & 1.6   & 0.291   \\
$n$=200   & 1.5    & 0.9   & 0.193  & 1.2    & 2.8   & 0.224 \\
$n$=500   & 1.9  &  1.1  &  0.133 &  1.7  &  2.9  & 0.155 \\
$n$=800   &  1.9  &  2.1 &  0.111 &   0.7 &   2.0  &   0.129 \\
$n$=1000   & 2.1 & 2.1 & 0.101 & 1.2  & 1.8  & 0.117 \\
\end{tabular}\\
\small{LR denotes the error rate that the true parameter is less than the lower confidence limit; UR denotes the error rate that the true parameter is greater than the upper confidence limit; WD is the average width of the confidence interval. The results of prior methods can be found in Tables 2-3 in the unpublished longer version of \cite{sen2007pseudolikelihood}.}
\end{center}
\end{table}

\subsection{Mean squared error of the point estimators}
In this section, we evaluate the mean squared error of the proposed fiducial point estimator of $F(t_0)$ for the above four scenarios. 
Furthermore, we compare it with the NPMLE estimator implemented in \cite{interval}. 
The default values of the parameters in the function \texttt{interval::icfit} are used. Moreover, the NPMLE estimator is not uniquely defined.
If there is not a unique NPMLE for a specific time, then we consider the following choices specified in \texttt{interval::getsurv}.

\begin{itemize}
\item Interpolation: take the
point on the line connecting the two points bounding the non-unique NPMLE interval;

\item Left: take the left side of the non-unique NPMLE interval (smallest $S(t)$, largest $F(t)$);

\item Right: take the right side of the
non-unique NPMLE interval (largest $S(t)$, smallest $F(t)$).
\end{itemize}

\begin{table}[H]
\begin{center}
{\textbf{Table~3}. Mean squared error ($\times 10^{-4}$) of point estimators for $F(t_0)$.}\\
\begin{tabular}{cccccccccccc}
         & \multicolumn{4}{c}{Scenario 1} & \multicolumn{4}{c}{Scenario 2} \\
         & F & MLE-I      & MLE-L     & MLE-R   & F & MLE-I      & MLE-L     & MLE-R      \\
$n$=50   &  103   &  225  &  236  & 246  &   116 & 272  &  295  &  281   \\
$n$=75   &   69  &   158   &  163   &  160  &  83 &  189  & 198  & 191 \\
$n$=100   &  56   &  127  &  132  &  133   & 69  & 152  &   158 & 156  \\
$n$=200   &  36   &  88  & 89  &  89   &  39  & 95 &  97  & 97 \\
$n$=500   &  17  &  44  & 45  &  45  &  19  &  50  &  50  & 50  \\
$n$=800   &   12  & 31  &  32 &  32  & 14 & 36  & 37  &  37 \\
$n$=1000   &  10 &   27   &  27 &  27 & 11 & 29  &  29  & 29  \\
\end{tabular}\\
\small{F denotes the proposed fiducial point estimator; MLE-I, MLE-L, and MLE-R denote the NPMLE with three specifications ``interpolation'', ``left'', and ``right'', respectively.}
\end{center}
\end{table}

\begin{table}[H]
\begin{center}
{\textbf{Table~4}. Mean squared error ($\times 10^{-4}$) of point estimators for $F(t_0)$.}\\
\begin{tabular}{cccccccccccc}
         & \multicolumn{4}{c}{Scenario 3} & \multicolumn{4}{c}{Scenario 4} \\
         & F & MLE-I      & MLE-L     & MLE-R   & F & MLE-I      & MLE-L     & MLE-R      \\
$n$=50   &   67  &  137  &  143  & 137  & 90  &  175 &   190  &   183  \\
$n$=75   &  42   &  86 &   88  &  89   &  66  & 124  & 129   & 129   \\
$n$=100   &  33   &   70 &  71  &  72   &  47  & 92  &  98  &  95 \\
$n$=200   &   20  &  45  &   46 &  46   &  28  &   56 &  58  &  57 \\
$n$=500   &  10  &  23  &  23  &  24  &  14  &  29  &  29   & 29  \\
$n$=800   &  7  & 16 & 16 &  16 &  9  &  21  & 21  &  21  \\
$n$=1000   &   6  &   14   & 14 & 14  & 7 & 15 &  15 & 15  \\
\end{tabular}\\
\small{F denotes the proposed fiducial point estimator; MLE-I, MLE-L, and MLE-R denote the NPMLE with three specifications ``interpolation'', ``left'', and ``right'', respectively.}
\end{center}
\end{table}

The mean squared errors are presented in Tables~3 and 4 for each scenario.
As sample size increases, all methods have higher estimation accuracy.
In addition, we see that the proposed fiducial approach has the smallest mean squared errors. Moreover, the mean squared errors of the NPMLE are twice as large as the fiducial estimator.
The observed patterns are consistent across all four scenarios and different sample sizes.

\section{Real data application}\label{sec:realdata}

\subsection{Current status data}

We consider a dataset on the prevalence of rubella in 230 Austrian males older than three months for whom
the exact date of birth was known \citep{keiding1996estimation}. Each individual was tested at the Institute of Virology, Vienna
during March 1-25, 1988, for immunization against Rubella. The
Austrian vaccination policy against Rubella at the time had long been to routinely
immunize girls just before puberty but not to vaccinate the boys, so that the male Austrians can represent an unvaccinated population.

Similar to \cite{keiding1996estimation,banerjee2005confidence}, our goal is to estimate the distribution of the time to infection (and subsequent
immunization) with rubella in the male population. It is assumed that immunization once
achieved, is lifelong.
\cite{keiding1996estimation,banerjee2005confidence} analyzed these data using the current status model. We apply the proposed fiducial approach to this dataset
with the range of observed times equally divided into 100 intervals as a fiducial grid. The fiducial estimates were based on 1000 iterations after 100 burn-in times.

As can be seen from Figure~\ref{pic:real1}, the estimated distribution function is similar to that of the NPMLE, as shown in Figure~1 of \cite{banerjee2005confidence}. The distribution function seems to rise steeply in the age range from 0 to 20 years. There is no significant change beyond 30 years, indicating that almost all individuals were immunized in their youth. 

\yifan{The shape of our 95\% confidence interval is similar to the likelihood ratio-based confidence intervals as presented in Figure~1 of  \cite{banerjee2005confidence}.
Figure~2 of \cite{banerjee2005confidence} shows the lengths of the confidence intervals, as a function of $t$ for the likelihood ratio-based confidence interval, parametric maximum likelihood based interval, non-parametric maximum likelihood based interval, and subsampling-based method. As stated in \cite{banerjee2005confidence}, ``none of the methods can be expected to come up with the shortest intervals in any given situation.''
However, the maximum length, taken over all $t$, of the proposed fiducial confidence intervals is 0.329. This appears to be much shorter than the maximum lengths reported in Figure~2 of \cite{banerjee2005confidence}.
}

\begin{figure}[h]
\begin{center}
\includegraphics[width=6.5cm]{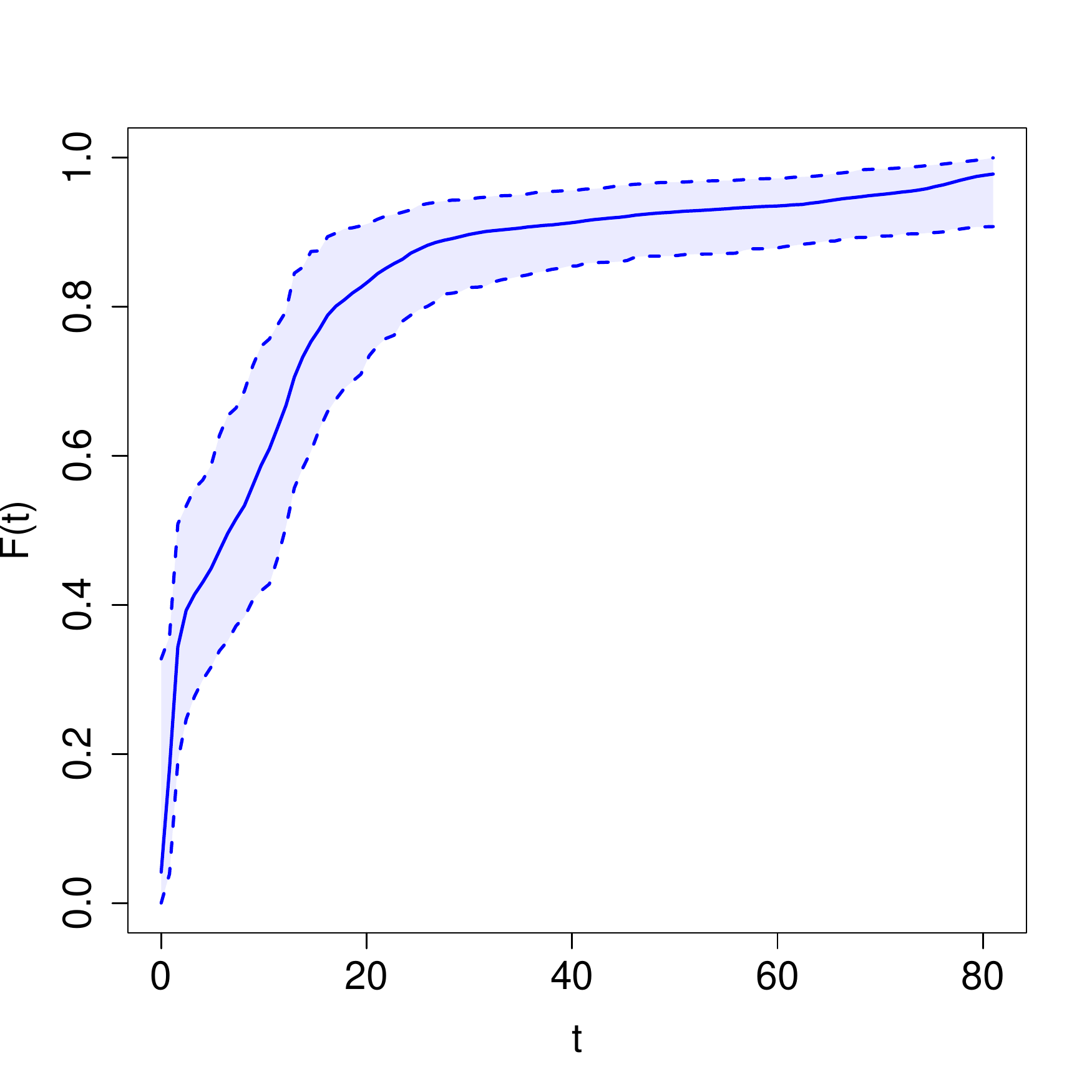}
\end{center}
\caption{Austrian rubella data: the estimated cumulative distribution function (solid line) and 95\% confidence interval (dashed lines) for $F(t)$.
}
\label{pic:real1}
\end{figure}

\subsection{Mixed case censoring data}
In this section, we consider a classic dataset given in \cite{de1989analysis} of a cohort study of haemophiliacs that were at risk of infection
with HIV.
 Since 1978, 262 people with hemophilia have been treated at Hopital Kremlin Bicetre and Hopital Coeur des Yvelines in France.
 The data consist of two groups:
105 patients in the heavily treated group, that is in the group of patients who received
at least 1000 $\mu g/kg$ of blood factor for at least one year between 1982 and 1985, and 157 patients in the lightly treated group, corresponding to those patients who received less than 1000 $\mu g/kg$ per year.
  By August 1988, 197 had  become infected, 97 in the heavily treated group and 100 in the lightly treated group, and 43 of these had developed clinical symptoms relating to their HIV infection. All of the infected persons are believed to have become infected by the contaminated blood factor they received for their hemophilia.

We are interested in estimating the distribution of time to HIV infection $T$ ($T = 1$ denotes July 1, 1978). The observations are based on a discretization of the time axis into six-month intervals. For each patient, the only information available is that $T \in (L,R]$.
We apply the proposed fiducial approach separately to the two different groups with the range of observed times equally divided into 100 intervals as a fiducial grid. The fiducial estimates were based on 1000 iterations after 100 burn-in times.

Due to the lack of information about the other inspection times, the full mixed case model cannot be fitted to the data.
\cite{sen2007pseudolikelihood} formulated the problem as a case II censoring model, which is a simplification for the purpose of illustrating their method; while for the proposed fiducial approach, we do not necessarily treat the data as a case II censoring problem due to the nature of our unified approach.

\begin{figure}[h]
\begin{center}
\includegraphics[width=6.5cm]{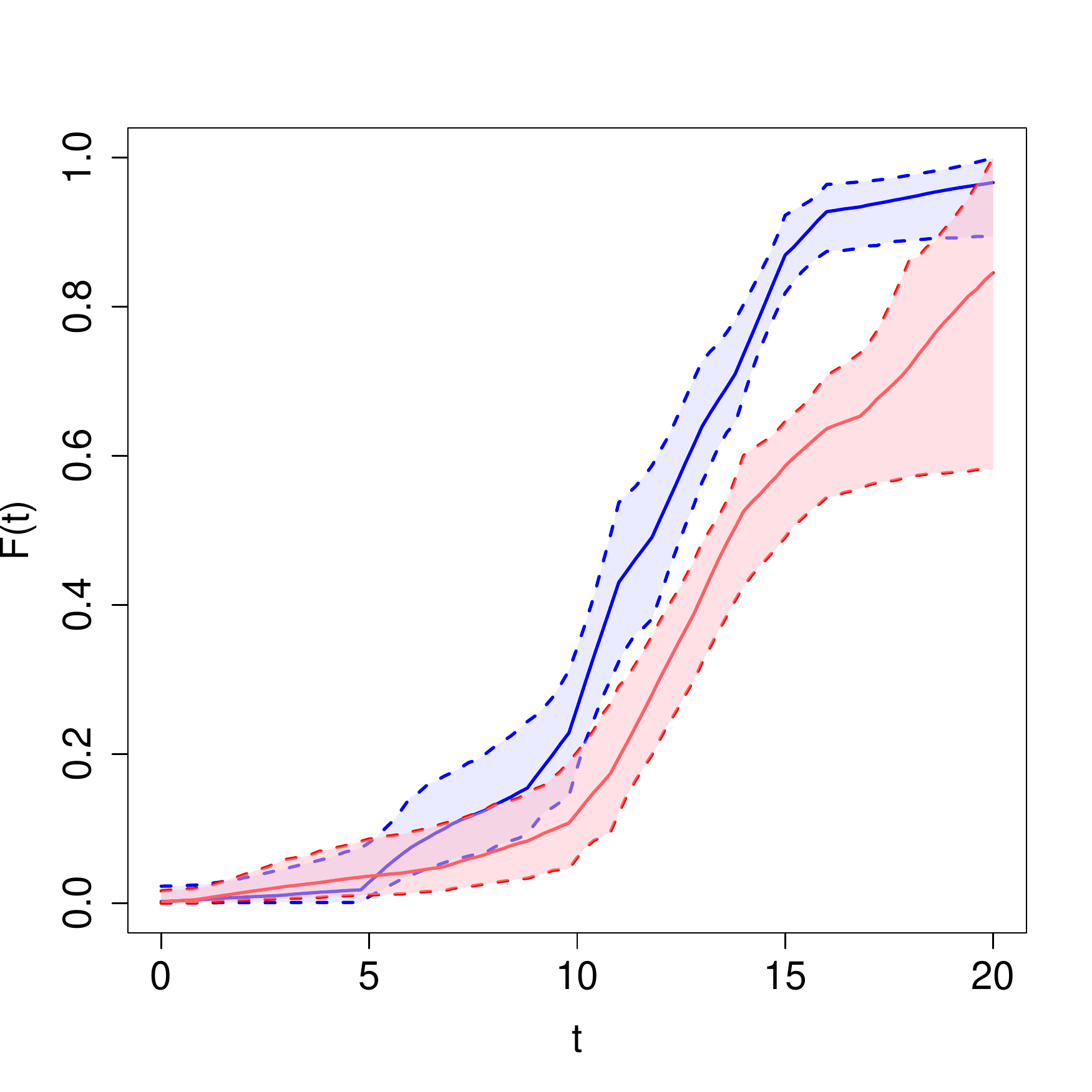}
\end{center}
\caption{A cohort study of hemophiliacs infected with HIV: the estimated cumulative distribution functions (solid lines) and 95\% confidence intervals (dashed lines) for $F(t)$  of two groups, respectively. The blue curves correspond to the heavily treated group; the red curves correspond to the lightly treated group. \yifan{Here, time is measured in 6-month intervals, with $t=1$ denoting July 1, 1978.}
}
\label{pic:real2}
\end{figure}

\yifan{In Figure~\ref{pic:real2} we see a sharp rise in the frequency of infections beginning around $t=9$, with
 infections for the heavily treated group occurring somewhat sooner. Such a difference
 in the two distributions has biological plausibility because heavily treated patients are likely
 to have received greater concentrations of HIV \citep{de1989analysis}.}
As can be seen from our Figure~\ref{pic:real2} as well as Figure~1 and Table~2 in \cite{sen2007pseudolikelihood}, although the overall trends across the two groups are similar among various methods, our results differ slightly from \cite{sen2007pseudolikelihood} in the range $(14, 16)$. The distribution function for the heavily treated group dominates the lightly treated group from day 6, while \cite{sen2007pseudolikelihood} found that, between 14 and 16, the distribution function for the lightly treated group is higher.
\yifan{Our findings are more in line with a nonparametric Bayesian approach \citep{calle2001nonparametric} as well as a self-consistency algorithm \citep{gomez1994estimation}}.

\newpage 
\appendix

\numberwithin{equation}{section}

\begin{center}
\textbf{\Large Appendix}
\end{center}

\section{Lemma~\ref{lemma:equiv2} and its proof}\label{sec:A}

\begin{lemma}
$Q_{\bl,\br}(\bu)\neq\emptyset$ if and only if $\bu$ satisfy: whenever $r_i\leq l_j$ then $u_i< u_j$.
\label{lemma:equiv2}
\end{lemma}
\begin{proof}
Sufficiency: If $Q_{\bl,\br}(\bu)\neq\emptyset$ holds, and $r_i\leq l_j$, then we know that $ u_i \leq F(r_i)\leq F(l_j)< u_j $.

Necessity: We prove this by contradiction. If $Q_{\bl,\br}(\bu)$ is empty, then there must exist indices $i$ and $j$ such that, $(l_j,r_j]$ is strictly larger than $(l_i,r_i]$ but $u_i \geq u_j$. This contradicts with whenever $r_i\leq l_j$ then $u_i< u_j$.
\end{proof}

\section{Proof of Proposition~\ref{theorem:consistency}}
\begin{proof}
\yifan{Let us first consider the current status data, $K=1$.} 
We adopt the notation of \cite{groeneboom1992information}, and denote $\delta_i=1$ if $l_i=0$ and $\delta_i=0$ if $r_i=\infty$.
Thus, the observed data $\{I_i=(l_i,r_i], i=1,\ldots, n\}$ can be recorded as $\{(x_i,\delta_i), i=1,\ldots,n\}$, where $x_i= r_i$ if $\delta_i=1$ and $x_i= l_i$ if $\delta_i=0$.
We have that the fiducial plausibility
\begin{align*}
\text{Pr}^\star(F\in Q_{\bl,\br}(\bU^\star)) & \yifan{\propto} \prod_{i=1}^n F(x_i)^{\delta_i}(1-F(x_i))^{1-\delta_i}\\
& = \exp\left\{ \sum_i \left[\delta_i \log(F(x_i)) + (1-\delta_i) \log(1-F(x_i)) \right ]\right\},
\end{align*}
where $\text{Pr}^\star$ denotes the fiducial probability.
Recall that the NPMLE solves the following optimization problem,
\begin{align*}
\max_{0\leq y_1,\ldots,y_n\leq 1} \sum_i \left[\delta_{i} \log(y_i) + (1-\delta_{i}) \log(1-y_i) \right ].
\end{align*}
\yifan{A detailed derivation of the closed form NPMLE can be found in \cite{sun2007statistical}.}
Thus, maximizing fiducial probability is equivalent to solving the optimization problem of the NPMLE estimator.

\yifan{In general, we have
\begin{align*}
\text{Pr}^\star(F\in Q_{\bl,\br}(\bU^\star))  \yifan{\propto} \prod_{i=1}^n [F(r_i)-F(l_i)],
\end{align*}
where $\text{Pr}^\star$ denotes the fiducial probability.
Therefore, any $F$ that maximizes the fiducial probability $\text{Pr}^\star(F\in Q_{\bl,\br}(\bU^\star))$ is an NPMLE estimator. 
Different algorithms for the NPMLE such as the self-consistency algorithm \citep{efron1967two,turnbull1976empirical,dempster1977maximum}, the iterative convex minorant algorithm \citep{groeneboom1992information,jongbloed1998iterative}, and a hybrid of self-consistency and iterative convex minorant algorithm \citep{wellner1997hybrid}, can be found in \cite{sun2007statistical}.}
\end{proof}

\section{Proof of Theorem~\ref{main}}\label{sec:D}
\begin{proof}
We need to study the distribution of $ Q_{\bl,\br}(\bU^\star)$
where $\bU^\star$ follows uniform distribution on the set
$\{\bu^\star: Q_{\bl,\br}(\bu^\star)\neq\emptyset\}$. Recall that we have $T_i=F^{-1}(U_i)$ and $L_i < F^{-1}(U_i) \leq R_i$  from Section~\ref{sec:method}.
Given Assumption~\ref{as:2}, we have
 \begin{equation}\label{eq:length1}
    n^2 \max_{i=1,\ldots,n} |r_i-l_i| \to 0 ~\text{in probability.}
 \end{equation}

We shall see that the unobserved $T_i=F^{-1}(U_i)$ are well separated.
Straightforward calculation with uniform order statistics shows that
\begin{align}\label{eq:length2}
\Pr\left(\min_{i\in\{0,\ldots,n\}} \left \{U_{(i+1)}-U_{(i)}\right\} >\frac{t}{n(n+1)}\right)\geq \left(1-\frac{t}{n}\right)^{n},
\end{align}
where $U_{(0)} \equiv 0$ and $U_{(n+1)} \equiv 1$. 
Equations~\eqref{eq:length1} and \eqref{eq:length2} together imply that
\[\Pr((l_i,r_i)\cap (l_j,r_j)\neq\emptyset \mbox{ ~for some~ $i\neq j$})\to 0.\]

\yifan{
Define
\begin{equation}\label{eq:OracleFidL}
\widetilde F(s)\equiv \sum_{i=0}^n I[
T_{(i)} \leq s< T_{(i+1)}]  U^*_{(i)}.
\end{equation}
Thus, we have that
\begin{align*}
 \text{Pr}(\sup_s n^{1/2} \{\widetilde F(s)- F^L(s) \}>\epsilon ) 
\leq & \sum_{i=0}^n \text{Pr}(U^*_{(i+1)}-U^*_{(i)} > \frac{\epsilon}{n^{1/2}})  \\ 
= & (n+1) \times \text{Pr}(Beta(1,n)>\frac{\epsilon}{n^{1/2}} )\\
= & (n+1) \times (1-\frac{\epsilon}{n^{1/2}})^n
\rightarrow  0.
\end{align*}
Therefore, \begin{align*}
\sup_s n^{1/2} |\widetilde F(s)- F^L(s) | \rightarrow 0,
\end{align*}
in probability. 
By Lemma~\ref{corollary}, we have that}
\begin{align*}
n^{1/2}\{ F^L(\cdot)-\widehat F(\cdot)  \} \rightarrow \{1-F(\cdot)\} W(\gamma(\cdot)),
\end{align*}
in distribution on Skorokhod space $D[0,\tau]$ in probability with $F(\tau)<1$,
where $W$ is the Brownian motion. Thus, for any $t<s$,
$$\operatorname{cov}[\{1-F(s)\}W(\gamma(s)),\{1-F(t)\} W(\gamma(t))]=\gamma(t)\{1-F(s)\}\{1-F(t)\}=F(t)\{1-F(s)\},$$ which completes the proof.

\end{proof}

\yifan{
\section{Lemma~\ref{corollary} and its proof}
\begin{lemma}\label{corollary}
Assume the conditions of Theorem~\ref{main}. We have 
\[
n^{1/2}\{ \widetilde F(\cdot)-\widehat F(\cdot)  \} \rightarrow \{1-F(\cdot)\} B(\gamma(\cdot)),
\]
in distribution on Skorokhod space $D[0,1]$ in probability, where $B$ is the Brownian motion, $\gamma(t)=\int_0^t \frac{f(s)}{[1-F(s)]^2} ds= \frac{F(t)}{1-F(t)}$, $\widehat F$ is defined in Theorem~\ref{main} as
 \begin{equation}\label{eq:OracleECDF}
\widehat F(s) \equiv \frac{1}{n}\sum_{i=1}^n I[T_i\leq s],
\end{equation}
and $\widetilde F$ is defined in Section~\ref{sec:D} as
\begin{equation*}
\widetilde F(s)\equiv \sum_{i=0}^n I[
T_{(i)} \leq s< T_{(i+1)}]  U^*_{(i)}.
\end{equation*}
\end{lemma}
\begin{proof}
By Theorem~2 of \cite{cuihannig2019}, we essentially need to check their Assumptions 1-3. Their Assumption~1 satisfies with their $\pi(t)= 1-F(t)$; their Assumption~2 satisfies as we assume true cumulative distribution function is absolutely continuous; their Assumption~3 satisfies as $$\int_0^t \frac{g_n(s)}{\sum_{i=1}^n I(T_i\geq s)} d [\sum_{i=1}^n I(T_i\leq s)] \rightarrow \int_0^t \frac{f(s)}{[1-F(s)]^2} ds,$$ for any $t$ such that $1-F(t)>0$ and any sequence of functions $g_n \rightarrow \frac{1}{1-F}$ uniformly.
\end{proof}
}

\yifan{
\section{Censoring mechanism}\label{app:censoring}
Consider the following data generating equation: 
\begin{align}\label{eq:dge}
    (L_i, R_i ]= G(V_i, \theta_i, T_i),\quad T_i=F^{-1}(U_i),\quad i=1,\ldots,n,
\end{align}
where $V_i, U_i$ are independent \tunif(0,1), and $G$ satisfies the following assumptions:
\begin{enumerate}[label=(\alph*)]
\item\label{asm:a}
$
    L_i <  T_i  \leq R_i,
$
for any $V_i$ and $\theta_i$;
\item\label{asm:b} for the observed $(L_i,R_i]$, any $T_i\in (L_i,R_i]$ and $V_i\in(0,1)$, there exists $\theta_i$ satisfying \eqref{eq:dge}.
\end{enumerate}
We assume that we only observe the intervals $(L_i,R_i]$, i.e., the true failure times $T_i$ are unobserved.}

\yifan{
The function $G$ determines the type of censoring and is assumed to be known. The unknown $\theta_i$ determines the censoring distribution and can be infinite dimensional. 
To demonstrate how \eqref{eq:dge} is used, we provide two classical censoring examples.
\begin{example}\textit{(Right-censoring)}
For the right-censored data, the function $G$ in \eqref{eq:dge} is defined as  follows: The unknown parameters $\theta_i$ are the distribution functions  $H_i$ of censoring times, and
\begin{align*}
    L_i= & T_i \wedge  H_i^{-1}(V_i | T_i = F^{-1}(U_i) ),\\
    R_i=& \begin{cases}
T_i, \text{~if~} T_i \leq H_i^{-1}(V_i | T_i = F^{-1}(U_i) );\\
\infty, \text{~if~} T_i > H_i^{-1}(V_i | T_i = F^{-1}(U_i) ).\\
\end{cases}
\end{align*}
The traditional right-censoring case observations $\{Y_i,\Delta_i\}$ then are $Y_i=L_i$ and $\Delta_i=I\{L_i=R_i\}$. Recall that, when $L_i=R_i=T_i$, \eqref{eq:dge} is modified to be non-empty.
\end{example}
\begin{example}\textit{(Case $K$ censoring)}
For the case $K$ censoring, the function $G$ in \eqref{eq:dge} is defined as  follows: The parameters are a collection of stochastically ordered distribution functions $\theta_i=\{H_{i,1},\ldots  H_{i,K}\}$, the inspection times
$C_{i,s}= H_{i,s}^{-1}(V_i | T_i = F^{-1}(U_i) )$, $i=1,\ldots,n$, $s=1,\ldots,K$, and 
\[
    L_i=  \sum_{s=0}^{K} C_{i,s} I\{C_{i,s} < T_i\leq   C_{i,s+1}\},\quad
    R_i= \sum_{s=0}^{K} C_{i,s+1} I\{C_{i,s} < T_i\leq   C_{i,s+1}\},
\]
where $C_{i,0}=0$ and $C_{i,K+1}=\infty$.
Note that, if all $C_{i,s}$ are deterministic, this can be viewed as rounded data. 
\end{example}
}

\yifan{
We now derive the generalized fiducial distribution based on this data generating equation
\eqref{eq:dge}. The corresponding inverse map for a single observation is:
\[
Q^{F,\theta_i}_{l_i,r_i}(u_i,v_i)
=\{F,\theta_i: F(l_i)< u_i\leq F(r_i), l_i <  G(v_i, \theta_i, F^{-1}(u_i)  ) \leq r_i\}.
\]
Consequently, the inverse map for the entire data is
 $Q^{F,\mathbf{\theta}}_{\bl,\br}(\bu,\bv)=\bigcap_i Q^{F,\theta_i}_{l_i,r_i}(u_i,v_i)$. 
Set 
\begin{equation*}
Q^F_{\bl,\br}(\bu) =
\{F\, :\ \, F(l_i)< u_i\leq F(r_i),\ i=1,\ldots,n\}.
\end{equation*}
Assumption~\ref{asm:b} implies that 
$Q^{F,\mathbf{\theta}}_{\bl,\br}(\bu,\bv)\neq\emptyset$ if and only if $Q^F_{\bl,\br}(\bu)\neq\emptyset$. Consequently, the marginal fiducial distribution is the distribution of $Q^F_{\bl,\br}(\bU^\star)$, where recall that the distribution of $\bU^\star$ is the uniform distribution on the set $\{\bu: Q^{F}_{\bl,\br}(\bu)\neq\emptyset\}$.}

\yifan{Note that this marginal fiducial distribution of $F$ agrees with the fiducial distribution derived in Section~\ref{sec:2.2}. In particular, it depends only on the observed values of $\bl,\br$ and not on the censoring mechanism. Therefore, as long as $G$ satisfies Assumptions~\ref{asm:a} and \ref{asm:b}, we do not need to know it.}

\bibliographystyle{asa}
\bibliography{fiducial}

\end{document}